\documentclass{article}
\usepackage[utf8]{inputenc}
\usepackage[T1]{fontenc}
\usepackage{algorithm}
\usepackage{algorithmic}
\usepackage{import}
\usepackage{url}
\usepackage{float} 
\usepackage[dvipsnames]{xcolor}
\usepackage{array}

\usepackage{amsmath, amsfonts, amssymb, amsthm}
\usepackage{mathtools}

\usepackage[colorlinks=true]{hyperref} 

\usepackage[hmargin=4.5cm, vmargin=2cm]{geometry}

\usepackage{xspace}

\usepackage{bm} 

\usepackage{graphicx}
\graphicspath{{./figures/}} 

\usepackage{parskip} 
\usepackage{comment}
\usepackage{fancyhdr}

\usepackage{tikz}
\usetikzlibrary{automata, positioning}
\usetikzlibrary{external}
\usepackage{csquotes}
\usepackage[backend=biber,citestyle=authoryear,bibstyle=alphabetic, backref]{biblatex}
\newtheorem{theoreme}{Theorem}
\newtheorem{lemma}[theoreme]{Lemma}

\newtheorem{definition}[theoreme]{Definition}

\newtheorem{remarque}[theoreme]{Remarque}

\newtheorem{corollary}[theoreme]{Corollary}

\newcommand{\ba}{\begin{aligned}}
\newcommand{\ea}{\end{aligned}}

\newcommand{\reeq}[1]{Equation~\eqref{eq:#1}}

\newcommand{\resec}[1]{Section~\ref{sec:#1}}
\newcommand{\redef}[1]{Definition~\ref{def:#1}}

\newcommand{\refig}[1]{Figure~\ref{fig:#1}}

\newcommand{\relem}[1]{Lemma~\ref{lem:#1}}
\newcommand{\recor}[1]{Corollary~\ref{cor:#1}}

\newcommand{\reapp}[1]{Appendix~\ref{app:#1}}

\newcommand{\gu}[1]{``#1''}	

\newcommand{\citet}[1]{\textcite{#1}}   
\renewcommand{\cite}[1]{\autocite{#1}}   

\DeclareMathOperator{\id}{Id}
\DeclarePairedDelimiter{\floor}{\lfloor}{\rfloor} 

\newcommand{\eps}{\varepsilon}

 \newcommand{\go}[1]{O\left( #1 \right)}

\newcommand{\paren}[1]{\left( #1 \right)}

\newcommand{\abs}[1]{\left\vert #1 \right\vert}

\newcommand{\inv}[1]{\frac{1}{#1}}

\newcommand{\nrm}[1]{\left\Vert #1 \right\Vert}

\newcommand{\nrmi}[1]{\left\Vert #1 \right\Vert_\infty}	

\newcommand{\bo}{B}

\newcommand{\enstq}[2]{\left\lbrace #1 \left\vert \, #2 \right. \right\rbrace}	

\newcommand{\cpl}[2]{\paren{#1,\,#2}}	
\newcommand{\tpl}[3]{\paren{#1,\,#2,\,#3}}	

\newcommand{\deq}{\mathrel{\mathop{:}}=}

\newcommand{\dist}[2]{\mathrm{d}\cpl{#1}{#2}} 

\newcommand{\erdosrenyi}{Erd\H{o}s-Rényi\xspace}
\newcommand{\brbalbert}{Barab\'asi-Albert\xspace}

\newcommand{\nods}{\mathcal{N}}	

\newcommand{\node}{n}	

\newcommand{\mdpl}{\mathrm{M}}	
\newcommand{\statiod}{\tilde{\mu}_\mdpl}    

\newcommand{\modl}{\mathcal{M}}  

\newcommand{\flot}{\Phi}    

\newcommand{\fix}{\mathrm{Fix}} 
\newcommand{\sglin}{\mathcal{H}}    
\newcommand{\aut}{\mathrm{Aut}} 

\newcommand{\tmax}{T_\infty}    
\newcommand{\timetrain}{T_\mathrm{train}} 

\newcommand{\obsmat}{\hat{O}} 

\newcommand{\mrcn}{\mdpl_{\mathrm{rec}}}    

\newcommand{\vecbeta}{\boldsymbol{\beta}} 
\newcommand{\vecdelta}{\boldsymbol{\delta}} 

\title{Network Reconstruction Problem for an Epidemic Reaction-Diffusion}

\date{2021}

\author{Louis-Brahim Beaufort, Pierre-Yves Massé, \\ Antonin Reboulet and Laurent Oudre}

\begin{document}

\maketitle

\begin{abstract}
We study the network reconstruction problem for an epidemic reaction-diffusion. These models are an extension of deterministic, compartmental models to a graph setting, where the reactions within the nodes are coupled by a diffusion. We study the influence of the diffusion rate, and the network topology, on the reconstruction and prediction problems, both from a theoretical and experimental standpoint.
Results first show that for almost every network, the reconstruction problem is identifiable. Then, we show that the faster the diffusion, the harder the reconstruction, but that increasing the sampling rate may help in this respect.
Second, we demonstrate that it is possible to classify symmetrical networks generating the same trajectories, and that the prediction problem can still be solved satisfyingly, even when the network topology makes exact reconstruction difficult.

\end{abstract}



\section{Introduction}
\label{sec:intro}

Network reconstruction problems, in which one aims at reconstructing a network structure from the observation of a signal evolving on it, is an important topic of current research, spanning over numerous domains \cite{timme14, shandi11, Dong2015LaplacianML, bars2019, sardellitti2019}.  Indeed, the widespread use of networks as a modelling tool in fields as diverse as telecommunications \cite{pastorsatorras2004, newman2002}, genetics \cite{gardner2003, karlebach2008}, ecology \cite{hanski1997, tamburello2019}, or transportation of goods or humans \cite{youn2008, perfido2017}, to name but a few, makes understanding the connections between their structure, or internal properties, and the phenomena which happen over them, a crucial issue.

Recently, 
\citet{prasse2020} have addressed the reconstruction problem for a wide class of epidemiological models \cite{sahneh2013}. These models have gained considerable attention since the early 20\textsuperscript{th} century, following notably the classic works of 
\citet{kermack}. In those, individuals are categorized in compartments which describe their status with respect to an infectious disease, and the models describe the way they transition from compartments to compartments as the disease spreads through contacts, and they react (heal) to it \cite{princeton}. Quickly, the early scalar models have been enhanced, by embedding them into networks \cite{pastor2001, pastor2015, nowzari15}, in order to refine the analysis of the influence of contacts between individuals, on the spread of the disease.

In their work, 
\citet{prasse2020} ask two questions: first, can the network structure be retrieved from the observation of the dynamics? Secondly, even in the case of a negative answer, is it possible to approximate the structure well enough to predict the future evolution of the disease? Even if the problem they study has a linear structure, they show the answers to these questions are not straightforward. We propose to address the same questions on another very important class of network-epidemiological models, the epidemic reaction-diffusion models, also known as metapopulation models with explicit movement \cite{arino09}. Just like the model studied in 
\citet{prasse2020}, the nodes of the graph represent sub-populations: for instance, the cities in the transportation network of a country. However, the interactions between populations is no longer described by a static contact structure, but by a diffusion. Accordingly, the internal dynamics of each sub-population follow a standard deterministic epidemiological model (SIS, SEIR, ...)  \cite{princeton}  while flows of individuals go from node to node through a diffusion. Following their apparition in population dynamics in the 1970's, these models have since gained considerable attention in the field of mathematical epidemiology \cite{brauer2001, vddriessche2002, wang2005, allen07, tien2015, arino17}. However, we are not aware the inverse problem has been studied for these models, up to now.

The standard network reconstruction procedure \cite{timme14} is the optimisation of some regression error. However, lack of identifiability, and bad conditioning, may prevent it from being highly efficient. Structural and qualitative analysis of the model is therefore of much importance to better understand the dynamics at hand, and guide the reconstruction work. Moreover, such analysis may provide insights for other models incorporating diffusion as well \cite{haehne2019}. Our contributions are therefore the following. On the one hand, we conduct theoretical analysis on the influence of the diffusion rate on the reconstruction, and illustrate our results by experiments. On the other hand, we study the influence of network topology, both theoretically, using notably the notion of graph automorphisms \cite{simon2011, ward2019}, and experimentally. Similar questions have been asked for other epidemic models \cite{Ganesh2005TheEO, Durrett4491, vajdi2018, prasse2021}.

We first present background material in \resec{background}. Next, we present the problem we address, conduct some initial identifiability analysis, and describe the experimental setup, in  \resec{problem_init_analysis}. Then, in \resec{inf_diff_rate}, we study the influence of the diffusion rate. Finally, in \resec{inf_topology}, we study the influence of the network topology. The proofs are deferred to the appendices. The code for the experiments, implemented in Python, is available on the git repository: \url{https://reine.cmla.ens-cachan.fr/masse/network_reconstruction_reaction_diffusion}.

\section{Background}
\label{sec:background}
We first present the classical epidemiological models (\resec{standard_models}), before giving a short overview on network reconstruction techniques (\resec{state_art_net_rec}). Then, we present our contributions (\resec{contributions}). Finally, we introduce our notations (\resec{notations}).

\subsection{Deterministic, Compartmental Epidemiological Models}
\label{sec:standard_models}
Deterministic, compartmental epidemiological models represent the propagation of a disease within a population by first segmenting the population in compartments, describing the status with respect to the disease \cite{princeton}. Classical compartments include the \gu{susceptible} (S), which gathers people which may contract the disease when confronted to \gu{infected} (I) people, who later will have \gu{recovered} (R). Transitions from compartments to compartments are governed by differential equations. One simple and generic model, which we use for simplicity throughout our study, is the SIR model. Three scalar functions $s$, $i$ and $r$ track the numbers of people in each compartment, and they evolve according to, for all $t\geq 0$,
\begin{equation}
\label{eq:sir-scalar}
\left\lbrace \ba
\frac{ds}{dt} &= - \beta s i \\
\frac{di}{dt} &= \beta s i - \delta i \\
\frac{dr}{dt} &= \delta i.
\ea \right.
\end{equation}
Here, $\beta$ and $\delta$ are positive real numbers. The parameter $\beta$ is often called the infection rate, and $\delta$ is the curing rate. The quantity $\delta^{-1}$ may be interpreted as the average time an individual remains infected, before healing \cite{MAT09}. The fact it is positive means people heal in finite time. It is well-known that the system of \reeq{sir-scalar} has a global solution for every initial condition $\tpl{s_0}{i_0}{r_0}$ with only nonnegative coordinates, and that solutions tend to equilibria of the form $\tpl{s_\infty}{0}{r_\infty}$ \cite{princeton}.

Works have extended these models to graphs, in order to increase their representative power \cite{nowzari15}. Nodes of the graphs represent either individuals, or sub-populations (cities, or countries, for instance). Accordingly, let us consider a possibly directed, (strongly, if directed) connected graph of size $N=\abs{\nods}$, where $\nods$ is the set of nodes $\node$.
In each node of the graph, a standard SIR reaction happens.
We write therefore $\beta_\node$ the infection rate, and $\delta_\node$ the curing rate, of node $\node$. Depending on the context, we write $\vecbeta$ (resp. $\vecdelta$) the diagonal matrix of coefficients $\beta_\node$ (resp. $\delta_\node$), or the vector $\paren{\beta_1,\,\ldots,\beta_N}$ (resp. $\paren{\delta_1,\,\ldots,\,\delta_N}$). Finally, we refer to the $\beta_\node$'s and $\delta_\node$'s as epidemiological parameters.
As we said in the introduction, we consider the model  where the internal node dynamics are coupled by a diffusion \cite{brauer2001, vddriessche2002, wang2005, allen07, tien2015, arino17}. It is governed by a diffusion matrix\footnote{We adopt this terminology, for lack of a universally agreed term for these matrices.}, which we define as follows.
\begin{definition}[Diffusion Matrix]
\label{def:diff_mat}
A diffusion matrix $\mdpl$ is first Metzler, that is for $i\neq j$, we have $\mdpl_{ij}\geq 0$. Then, it is irreducible\footnote{This is possible if the graph is directed because we ask it is then strongly connected.}. Thirdly, its columns have vanishing sums.
\end{definition}
The resulting reaction-diffusion dynamics is given by
\begin{equation}
\label{eq:sir-metapop}
\left\lbrace \ba
\frac{dS}{dt} &= - \vecbeta S \odot I + \mdpl S \\
\frac{dI}{dt} &= \vecbeta S \odot I - \vecdelta I + \mdpl I \\
\frac{dR}{dt} &= \vecdelta I + \mdpl R,
\ea \right.
\end{equation}
where, as we explain below in \resec{notations}, $\odot$ represents the coordinate-wise product\footnote{For instance, for a node $\node$, the equation on $S_\node$ reads: $d S_\node/dt = -\beta_\node S_\node(t)I_\node(t) + \sum_{i=1}^N \mdpl_{\node, i} S_i(t)$. Our notation is not standard, but we think it makes clearer the link of the graph system of \reeq{sir-metapop} with the scalar system of \reeq{sir-scalar}.}. Standard results guarantee that the solution to \reeq{sir-metapop} is global, and converges to a fix point of the form $\tpl{S}{0}{R}$, as $t\to\infty$ \cite{arino09}. Moreover, the total population is preserved, that is $\sum_\node S_\node(t)+I_\node(t)+R_\node(t)$ is constant.
Finally, standard Perron-Frobenius theory \cite{meyer2000} shows a diffusion matrix admits a unique stationary distribution, that is a positive vector $\statiod$ summing to $1$ such that $\mdpl\statiod=0$. Moreover, solutions of $dX/dt = \mdpl X$ with initial condition $X_0$ having a nonzero coordinate along $\statiod$ converge to $\statiod$, as $t\to \infty$. In particular, since the total population $S+I+R$ satisfies this equation, and its initial condition has a nonzero coordinate along $\statiod$ (one quickly checks it equals $\sum_\node S_\node(0)+I_\node(0)+R_\node(0)$, which is nonzero as $S(0)$, $I(0)$ and $R(0)$ have nonnegative coordinates), it converges to this stationary distribution, as $t\to\infty$.

\subsection{Background on Issues in Network Reconstruction}
\label{sec:state_art_net_rec}
The network reconstruction problem from observations, where one aims at expliciting the topology of a network of $N$ nodes, by observing the values taken by some dynamical system which evolves on it, has been extensively studied in the literature (see for instance the review \cite{timme14}). The network is described by some matrix $\mdpl$ (typically, the adjacency matrix, possibly weighted). Observations are often gathered in two matrices, $\hat Y$ and $\obsmat$, which typically belong to $\in\mathcal{M}_{N\times K}(\mathbb{R})$, where $K$ is the number of measurements. Often, $\hat Y$ gathers estimates of the time derivatives of the state of the dynamical system in each node, at the different measurement times, and $\obsmat$ is the so-called observation matrix gathering the values in each node, at the same times. Then, one knows the relation  $\hat Y = \mdpl \obsmat$ must be satisfied. Therefore, the problem amounts to solving this regression equation. We first describe the different observations possible, then address the solving of the regression.

Observations may first consist in measurements of the answer the system gives to some user-driven perturbation of its dynamics \cite{gardner2003, yeung2002, yu2010}. In the case of non linear dynamics, these perturbations may occur near a fixed point, the interest being that the first-order expansion of this system then depends linearly on the network \cite{gardner2003}, so that the observations $\hat Y$, $\obsmat$, and the network matrix $\mdpl$, indeed satisfy the $\hat Y = \mdpl \obsmat$ equation. Alternatively, observations may be obtained through mere observation of the system \cite{shandi11, makarov2005}.
The nature of $\hat Y$ and $\obsmat$ moreover depends on whether a model for the dynamics studied is known \cite{shandi11, gardner2003, wang2011, prasse2020}, or not \cite{quinn2011, barzel2013, mangan2016, casadiego2017}.  For instance, in 
\citet{vanbussel2011}, the authors use detailed knowledge of the evolution of a synthetic model of a biological synaptic network between spiking times, to obtain the matrices $\hat Y$ and $\obsmat$ verifying the $\hat Y=\mdpl \obsmat$ equation. 
On the other hand, \cite{casadiego2017} only assume some very general relation between the first order derivatives of the dynamical system, and the values it takes, in order to obtain similar relations.

Once obtained the observations such that the equation $\hat Y=\mdpl \obsmat$ holds, one must then solve the regression problem. It may be over-determined, if $K > N$, or under determined, if $K < N$ \cite{stoer1993}. Even if $X$ has rank $N$, it may be ill-conditioned, thus preventing efficient solving by mere matrix inversion. To address these issues, a standard choice is to minimise the regression error with respect to some norm. One choice is then between $L^1$ or $L^2$ (least-squares) optimisation. The former induces sparsity, which may be desirable. For instance, 
\citet{mangan2016} assume the dynamics decompose in some well-chosen basis, and that most of the coefficients in the expansion vanish. They then identify a subspace to which the vector of coefficients belongs, and finally use standard algorithms to find the sparsest vector in this subspace.
In 
\citet{yeung2002}, the authors use an SVD decomposition of some observation matrix to parametrize the set of networks consistent with the data, and then use sparse regression to find the sparsest such network.
In 
\citet{wang2011}, the authors decompose the dynamics over some infinite basis, then use compressed sensing to evaluate the coefficients, only few of them are then nonzero.

Least-square optimization is on the other hand less costly, and better suited for over-determined systems.
In 
\citet{prasse2020}, the authors use a least-square optimisation, but add a $L^1$ penalty in order to enforce some degree of sparsity, thus solving:
\begin{equation*}
\min_{\mdpl \in \mathcal{M}_N(\mathbb{R}^N)} \nrm{\hat{Y} - \mdpl \hat{O}}_2
  + \lambda \, \nrm{\mdpl}_1,
\end{equation*}
where $\lambda$ is selected by cross-validation.

Finally, let us also mention 
\citet{tyrcha2014}, which in another vein differentiates the dynamics of the model, in order to train it to reproduce the observations, as is usual for Recurrent Neural Networks.

\subsection{Contributions of the article}
\label{sec:contributions}

In our work, we  study the reconstruction, and prediction, problems, for an epidemic reaction-diffusion. We assume known a model, and we consider that observations are a given, standalone time-series, which is arguably the harder setting observations wise, and which seems more relevant in the case of epidemic dynamics.
We first show that for almost every network, the reconstruction problem is identifiable (\relem{traj_oft_whole_space}).
Then, we show that the quicker the diffusion, the lower the numerical rank of the observation matrix (\recor{rank_diff_rate_infty}), and the harder the reconstruction, but that increasing sampling helps reconstruct the network.
Then, we classify symmetrical networks generating the same trajectories (\relem{classification}). Finally, we show experimentally, on synthetic data constructed with random graph generators exhibiting different topologies, that reconstruction is easier for more \gu{constrained} topologies, and that the prediction problem can still be solved satisfyingly even when the network topology makes exact reconstruction difficult.
We use least-squares under constraints to solve numerically the reconstruction problem (see \resec{exp_setup}), in the experiments.

\subsection{Notations and main definitions}
\label{sec:notations}
We use the capital letter $X$ to designate vectors on $\mathbb{R}^{3N}$, for some integer $N\geq 1$, which write $X=\tpl{S}{I}{R}$, with $S, I, R\in\mathbb{R}^N$. Lower case $x$ designates vectors on $\mathbb{R}^3$, with $x=\tpl{s}{i}{r}$, and $s, i$ and $r$ real numbers. Whenever we consider some function $f$ defined over $\mathbb{R}^N$, we choose to extend the notation in a straightforward way to $\mathbb{R}^{3N}$, by writing, for $X=\tpl{S}{I}{R}$ as above, $f(X) \deq \tpl{f(S)}{f(I)}{f(R)}$. Whenever $X$ is a function describing a trajectory of a dynamical system, $X : \mathbb{R}_+ \to \mathbb{R}^{3N}$, we write $Y=f(X)$ the function defined by, for all $t\geq 0$, $Y(t)\deq f(X(t))$. As a result, combining with what precedes, for $X=\tpl{S}{I}{R}$, for all $t\geq 0$, we have $Y(t)=\tpl{S(t)}{I(t)}{R(t)}$. For a set $\mathcal{S}\subset\mathbb{R}^N$, for $X=\tpl{S}{I}{R}\in\mathbb{R}^{3N}$, we write $X\in\mathcal{S}$ to mean that $S$, $I$ and $R$ belong to $\mathcal{S}$.
Then, for two vectors $u$ and $v$ of equal dimensions, we write $u\odot v$ their coordinate wise product: $u\odot v$ is the vector $\paren{u_i v_i}_i$. 
Finally, $1_N$ designates the vector of $\mathbb{R}^N$ will all coordinates equal to $1$.
Let us finally introduce the following three definitions, which help us formalise the setting. 
\begin{definition}[Model]
We call model, and write $\modl=\tpl{\mdpl}{\cpl{\vecbeta}{\vecdelta}}{X_0}$, a tuple consisting of a diffusion matrix $\mdpl$, epidemiological parameters gathered in $\vecbeta$ and $\vecdelta$, and an initial condition $X_0=\tpl{S_0}{I_0}{R_0}\in\mathbb{R}^{3N}$.
\end{definition}
\begin{definition}[Flow on a model]
\label{def:flow_mod}
A flow $\flot$ on the set of models $\mathfrak{M}$ is a mapping:
\begin{equation*}
\ba
\flot & :  & \mathfrak{M} & \to & \mathcal{C}\cpl{\mathbb{R}_+}{\mathbb{R}^{3N}} \\
& & \modl & \mapsto & \flot\paren{\modl}
\ea
\end{equation*}
defined by, for each model $\modl=\tpl{\mdpl}{\cpl{\vecbeta}{\vecdelta}}{X_0}$, for all $t\geq 0$, $\flot_t\paren{\modl}$ is the value at time $t$ of the solution of the differential \reeq{sir-metapop}, with initial condition $X_0$, that is $\flot_t\paren{\modl}=\tpl{S(t)}{I(t)}{R(t)}$.
\end{definition}
Slightly abusing notations, in the following, we sometimes write $\flot(\mdpl)$ when $\cpl{\vecbeta}{\vecdelta}$ and $X_0$ are fixed, so that the model only depends on the choice of the diffusion matrix.

Let us finally introduce the observation matrix  \cite{tyrcha2014}, and the vectors of estimates of the reaction terms, and the derivatives.
We do not observe the whole trajectories, but only some samples of them. For some integer $K \geq 1$, let us then consider the sampling times $0=t_0 < t_1 < t_2 < \ldots < t_K$.
Let for each node $\node$ $\hat S_\node(t_k)$ be the (possibly noisy) observation of compartment $S$ in node $\node$ at time $t_k$, (and likewise for the other compartments).
We can also estimate the vectors of derivatives, and of reaction terms, of \reeq{sir-metapop}, from the observations, as is done in 
\citet{shandi11}. For every $1 \leq k \leq K$, we define $\hat \rho_S(t_k) = \vecbeta \hat S(t_k) \cdot \hat I(t_k)$ the vector of reaction terms on $S$ at time $t_k$, and $\hat D_S(t_k) = \paren{\hat S(t_{k})-\hat S(t_{k-1})}(t_{k}-t_{k-1})^{-1}$ the estimate of the derivative on S at time $t_k$. We do likewise for the other compartments.

\begin{definition}[Observation matrix, derivatives and reaction terms]
\label{def:observation_matrix}
Let the observation matrix on $S$ be
\begin{equation*}
\hat O_S =
\begin{pmatrix}
\hat S_1(t_1) & \ldots & \hat S_1(t_K) \\
\vdots & \ddots & \vdots \\
\hat S_N(t_1) & \ldots & \hat S_N(t_K)
\end{pmatrix} \in \mathcal{M}_{N\times K}\paren{\mathbb{R}}.
\end{equation*}
Note likewise $\hat O_I$ the observations on $I$, and $\hat O_R$ those on $R$.
Define finally the matrix by block $\hat O\cpl{\paren{t_k}}{\flot(\mdpl)} = \tpl{\hat O_S}{\hat O_I}{\hat O_R}\in\mathcal{M}_{N \times 3 K}\paren{\mathbb{R}}$. This is the observation matrix associated with the sampling times $\paren{t_k}$, and the flow $\flot(\mdpl)$.

Likewise, we write $\hat \rho_S, \hat D_S \in\mathcal{M}_{N\times K}\paren{\mathbb{R}}$ the matrices of reaction terms (resp. derivatives) on $S$, and likewise for the other compartments. We finally define $\hat{\rho}=\tpl{\hat \rho_S}{\hat \rho_I}{\hat \rho_R}\in\mathcal{M}_{N \times 3 K}\paren{\mathbb{R}}$ and $\hat D=\tpl{\hat D_S}{\hat D_I}{\hat D_R}\in\mathcal{M}_{N \times 3 K}\paren{\mathbb{R}}$.
\end{definition}

\section{Problem Studied, Identifiability and Experimental Setup}
\label{sec:problem_init_analysis}
We present the problem we address (\resec{pb_studied}), then initate the study of its identifiability (\resec{identifiability}), and finally present the setting in which we conduct the experiments (\resec{exp_setup}).
\subsection{Problematic: Reconstruction and Prediction}
\label{sec:pb_studied}
We assume known the initial condition $X_0$, and the epidemiological parameters $\beta_\node$'s and $\delta_\node$'s. Some unknown diffusion matrix $\mdpl^*$ then generates a flow $\tpl{S}{I}{R}=\flot(\mdpl^*)$, and given as observations the trajectories $S$, $I$ and $R$, we address the following two questions.
\begin{enumerate}
\item Question 1. Can we estimate $\mdpl^*$ from the observations? In other words, do they first uniquely define $\mdpl^*$? And in so, is it possible to estimate $\mdpl^*$ from them?
\item Question 2. Can we predict the future evolution of the system if, for some $\timetrain < \infty$, we can only observe the trajectories in some initial phase $0 \leq t \leq \timetrain$ of the system, that is the observations only consist in $\cpl{\flot_t(\mdpl^*)}{t\leq \timetrain}$? 
\end{enumerate}

\subsection{Identifiability of the Diffusion}
\label{sec:identifiability}

We now conduct some preliminary analysis on the identifiability of the diffusion.
We have the following characterisation of the set of diffusion matrices $\mdpl$ which produce the same trajectories as $\mdpl^*$ (see \reapp{model_problems} for a proof, and likewise for future results).
\begin{lemma}[Diffusions Generating the Same Trajectories]
\label{lem:set_same_traj}
Let $\mdpl^*$ be a diffusion matrix, and write $\tpl{S}{I}{R}=\flot\paren{\mdpl^*}$.
Then, every matrix $\mdpl=\mdpl^* + H$, such that
first $\mdpl$ is a diffusion matrix, and secondly such that for all $t\geq 0$, we have\footnote{As explained in \resec{notations}, $\flot_t(\mdpl)\in\ker{H}$ means that $S(t)$, $I(t)$ and $R(t)$ belong to $\ker{H}$.} $\flot_t(\mdpl)\in\ker{H}$,
produces the same trajectories as $\mdpl^*$.
\end{lemma}
As a result, provided the vector space generated by the trajectories, that is by the vectors $S(t)$, $I(t)$ and $R(t)$, for $t\geq 0$, is the whole space $\mathbb{R}^N$, then the answer to our first question is affirmative (as the only $H$ possible vanishes over the whole space, therefore vanishes). Therefore, a fundamental question governing the issue of the identifiability of the diffusion matrix is the existence of strict subspaces of $\mathbb{R}^N$ in which the trajectories evolve.
This moreover gives us a practical criterion to evaluate if the diffusion matrix generating a given trajectory is unique: we check if the observation matrix has rank $N$, which is sufficient to guarantee the uniqueness.
Now, often the trajectories generate the whole space, as the next result shows.
\begin{lemma}[Almost Everywhere Identifiability] 
\label{lem:traj_oft_whole_space}
Let $0 \leq t_1 < ... < t_N < \infty$ be a subdivision of the nonnegative real half-axis. 
Then, for almost every $\mdpl$, $X_0$, for all $\vecbeta$, $\vecdelta$, writing $\modl=\tpl{\mdpl}{\cpl{\vecbeta}{\vecdelta}}{X_0}$, the space generated by the samples of the trajectories at instants $t_1, ..., t_N$ (that is $\flot_{t_1}(\modl),\ldots, \flot_{t_N}(\modl)$), is equal to $\mathbb{R}^N$.
\end{lemma}

This might give the impression the problem is solved, for almost every $\mdpl$ and $X_0$. Indeed, assume the space generated by the trajectories is the whole of $\mathbb{R}^N$. Then, the observation matrix $\obsmat\cpl{\paren{t_k}}{\flot(\mdpl)}$ has rank $N$. We therefore know the image of $\mdpl$ on a basis, which fully determines it.
However, the conditioning of the observation matrix is often poor in practise, so that reconstruction of $\mdpl^*$ by extracting a basis is inefficient. In the next two sections, we investigate two reasons why this is the case. Firstly, we study the influence of the speed of diffusion in \resec{inf_diff_rate}. Secondly, we consider the topology of the graph in \resec{inf_topology}. Nonetheless, the fact the diffusion is often unique means that, when running a reconstruction algorithm, we can have good hope it will succeed in finding a good fit, which we show is the case in the experiments.

\subsection{Experimental Set-Up}
\label{sec:exp_setup}
For each experiment, we start by generating a graph from a random graph generator. We use four random graphs, with different topologies: the \erdosrenyi and the Waxman graph, which are quite connected graphs, and therefore quite \gu{close} to a complete graph, and the Relaxed Caveman and extended \brbalbert graphs, which are less connected, and exhibit a more clustered structure. In that sense, the diffusion is more constrained by these graphs, and we expect the reconstruction problems to be easier in that case. We study graphs of sizes ranging from $N=20$ to $N=140$. Then, we draw the coefficients of the diffusion matrix $\mdpl^*$ uniformly on $[0,1]$. Finally, when needed, we renormalise the diffusion matrix by the typical time of diffusion $\tau >0$. The epidemiological coefficients $\beta_\node$ and $\delta_\node$ are drawn at random, such that for each node $\node$, we have $\beta_\node \sim \abs{\mathcal{N}\cpl{1}{2}}$ and $\delta_\node \sim \abs{\mathcal{N}\cpl{3\times 10^{-2}}{6\times10^{-2}}}$.

Next, we simulate the ground truth trajectories on the time interval $[0, \tmax]$, with $\tmax=10$. We use a uniform time discretisation step of $\paren{\Delta t}_\mathrm{ground truth} = 10^{-3}$, and a Runge-Kutta discretisation scheme of order 4.
For each setting, we repeat the experiments at least 5 times, so as to control the stochastic fluctuations.
Finally, the train set in which we sample the observations is $[0, \timetrain]$, with $\timetrain=2$. We use various sample steps $\paren{\Delta t}_\mathrm{sample}$, ranging between $2\times 10^{-3}$ and $10^{-2}$.
Recall from \resec{standard_models} that we write $\statiod$ the stationary distribution of the diffusion matrix $\mdpl$. In each experiment, we use as initial condition $X_0 = \tpl{S(0)}{I(0)}{R(0)}= \tpl{s_0\statiod}{i_0\statiod}{r_0\statiod}$, where $s_0$, $i_0$ and $r_0$ are nonnegative real numbers, and $i_0 >0$. As a result, the vector of initial susceptibles $S(0)$ is proportional to the stationary distribution, and likewise for $I(0)$ and $R(0)$.
We compute the reconstructed diffusion matrix by solving
\begin{equation*}
\ba
\min_{\mdpl \in \mathcal{M}_N(\mathbb{R})} &\nrm{\hat{D} - \hat{R} - \mdpl \hat{O}}_2, \\
\text{such that} &\left\lbrace \ba
\mdpl_{i,j} &\geq 0, \, i\neq j, \\
\sum_j \mdpl_{i,j} &= 0 \quad \text{for all node} \, i.
\ea \right.
\ea
\end{equation*}
This is a convex optimisation problem.
We solved it using the Python package CVXPY \cite{diamond2016cvxpy, agrawal2018rewriting}. We write $\mrcn$ the matrix obtained.
Moreover, to truly enforce the fact $\mrcn$ is a diffusion matrix, we post-processed the matrix obtained by enforcing that column sums vanish: for every node $\node$, we replaced the diagonal coefficient $\mrcn\cpl{\node}{\node}$ by $-\sum_{i=1}^N \mrcn\cpl{i}{\node}$.

To assess the reconstruction, we use two metrics.
First, we use the AUC \cite{FAWCETT2006861} on the presence of edges, as \citet{prasse2020}. It is computed thanks to the corresponding fonction in Scikit-learn \cite{scikit-learn}. 
Secondly, we evaluate the prediction error, that is the norm of the difference between the trajectories computed with the true model, and those computed with the reconstructed diffusion $\mdpl_{\mathrm{rec}}$, by 
\begin{equation}
\label{eq:pred-error}
\inv{N} \inv{\tmax - \timetrain} \sum_{p=1}^{p_\mathrm{max}} \nrm{\flot_{t_p}(\mdpl^*) - \flot_{t_p}(\mdpl_\mathrm{rec})}^2 \paren{t_p-t_{p-1}},
\end{equation}
where $p_\mathrm{max}=\floor{\frac{\tmax-\timetrain}{\paren{\Delta t}_\mathrm{grountruth}}}$ and $\paren{t_p}$ is the discretisation scheme used for the simulations, whose beginning has been removed, so that $t_0=\timetrain$.

Computations with the \brbalbert graph were proner to numerical instabilities. We believe this is due to its topology being more constrained. As a result, we modified a bit the experimental setting for this specific graph, increasing sampling to $\paren{\Delta t}_\mathrm{sample} = 4\times 10^{-3}$,
and increasing the number of repetitions to $30$.

\section{Influence of the Diffusion Rate}
\label{sec:inf_diff_rate}
We now study the influence of the diffusion rate on the feasibility of the network reconstruction, first theoretically (\resec{diff_rate_analysis}), then experimentally (\resec{rcn_diff_rate_sampling}).

\subsection{Analysis}
\label{sec:diff_rate_analysis}
One difficulty of the network reconstruction problem is the conditioning of the observation matrix (\redef{observation_matrix}), which may be poor.
In particular, its numerical rank 
may be significantly lower than $N$, as observed also in \cite{prasse2020}. In our case, this may be partly due to the homogenisation performed by the diffusion. Indeed, given different epidemiological parameters, and different population sizes, the internal dynamics of the different nodes evolve differently. However, the diffusion tends to homogeneise each compartment, so that $S(t)$ tends to a vector proportional to the stationary distribution, $\statiod$, and likewise for $I(t)$ and $R(t)$. As a result, the diffusion tends to worsen the conditioning of a basis. This effect depends on the time-scale at which diffusion occurs, with respect to that at which the reactions in each node occur. We first show, in the following \relem{limit_trajs_diffusion_rate_infty}, that when the typical time of evolution of the diffusion, $\tau$, goes to $0$ (equivalently, the diffusion rate $1/\tau$ goes to infinity), and in the presence of fixed epidemiological parameters, the trajectories tend to those of a scalar SIR systems, which coefficients we express in terms of the $\beta_\node$'s, the $\delta_\node$'s and the stationary distribution, times the stationary distribution for each compartment. We then address the consequences for the numerical rank in \recor{rank_diff_rate_infty}.
For any $\tau >0$, we write $\flot^\tau=\flot\paren{\frac{\mdpl}{\tau}}$, that is the flow obtained by replacing $\mdpl$ by $\mdpl/\tau$ in \reeq{sir-metapop}.
\begin{lemma}[Limit Trajectories for Diffusion Rate going to Infinity]
\label{lem:limit_trajs_diffusion_rate_infty}
Let $\modl=\tpl{\mdpl}{\cpl{\vecbeta}{\vecdelta}}{X_0}$ be a model, and assume the initial condition $X_0$ is such that $S_0$, $I_0$ and $R_0$ are proportional to the stationary distribution $\statiod$. 
Write $\tpl{s}{i}{r}$ the solutions of the scalar system
\begin{equation*}
     \left\{
     \begin{aligned}
\frac{ds}{dt} &= - \tilde\beta si  \\
\frac{di}{dt} &= \tilde\beta si - \tilde\delta i  \\
\frac{dr}{dt} &= \tilde\delta i,
     \end{aligned}
     \right. 
\end{equation*}
with $s(0)=\sum_\node S_\node(0)$, and likewise for $i$ and $r$, and with
\begin{equation*}
\tilde\beta = \sum_\node \beta_\node \statiod(\node)^2, \quad \text{and} \quad
\tilde\delta = \sum_\node \delta_\node \statiod(\node).
\end{equation*}
Then, for any $T>0$, $\flot^\tau(\modl) \to \tpl{s\statiod}{i\statiod}{r\statiod}$, as $\to 0$, uniformly on $[0,T]$.
\end{lemma}
We illustrate \relem{limit_trajs_diffusion_rate_infty} on \refig{discrepancy_scalar_model}. We ran experiments according to the protocol described in \resec{exp_setup}, using \erdosrenyi and Relaxed Caveman graphs, for a range of values of $\tau$. For $\paren{t_p}$ the discretisation scheme used for the simulations, and $p_\mathrm{max}$ the number of $t_p$'s, we plot the error 
\begin{equation*}
\inv{N} \inv{\tmax} \sum_{p=1}^{p_\mathrm{max}} \nrm{\flot_{t_p}(\frac{\mdpl}{\tau}) - \tpl{s(t_p)\statiod}{i(t_p)\statiod}{r(t_p)\statiod}}_2^2 \paren{t_p-t_{p-1}}
\end{equation*}
between the trajectories obtained with the vector model, and those computed from the scalar model. We indeed see it goes to $0$, as $\tau \to 0$. Moreover, the discrepancy is bigger for the Relaxed Caveman graph, than for the \erdosrenyi one: indeed, the latter is more connected, and therefore there are much more exchanges between the nodes, so that it is closer to a kind of \gu{average} model, which the scalar limit is. 
\begin{figure}
    \centering
    \includegraphics[scale=.45]{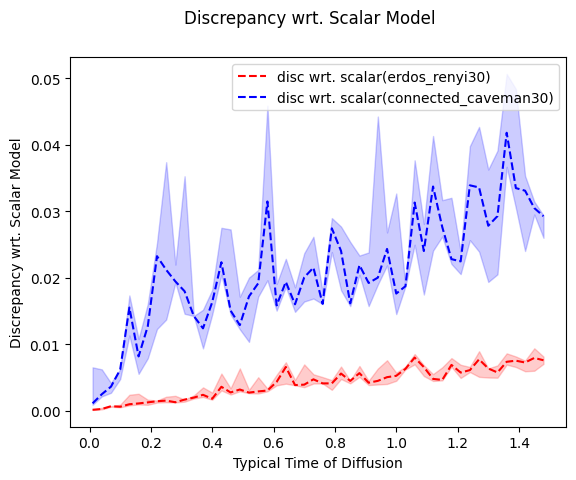}
    \caption{Discrepancy between the Scalar Model and the True Model as a Function of the Typical Time of Diffusion}
    \label{fig:discrepancy_scalar_model}
\end{figure}
Then, from \relem{limit_trajs_diffusion_rate_infty}, 
we immediately have the following corollary which describes its consequences for the numerical rank of the observation matrix.
\begin{corollary}[Numerical Rank of the Observation Matrix for Diffusion Rate going to Infinity]
\label{cor:rank_diff_rate_infty}
We make the same assumptions as in \relem{limit_trajs_diffusion_rate_infty}.
Let, for some integer $K\geq 1$, $\paren{t_k}_{1\leq k \leq K}$ be a family of sample times. Let $\tau> 0$, and let us write $\obsmat\cpl{\paren{t_k}}{\flot^\tau(\mdpl)}$ the observation matrix associated with the $t_k$'s, and the flow $\flot^\tau(\mdpl)$. Then, the numerical rank of the matrix $\obsmat\cpl{\paren{t_k}}{\flot^\tau(\mdpl)}$ goes to $1$, as $\tau\to 0$.
\end{corollary}
We illustrate this convergence on \refig{rank_diff_rate_infty}. We use the same protocol as for \refig{discrepancy_scalar_model}, but this time display the numerical rank. We see it gets lower and lower, as $\tau\to 0$. It is lower for the \erdosrenyi graph, probably for the same reasons given above. 
\begin{figure}
    \centering
    \includegraphics[scale=.45]{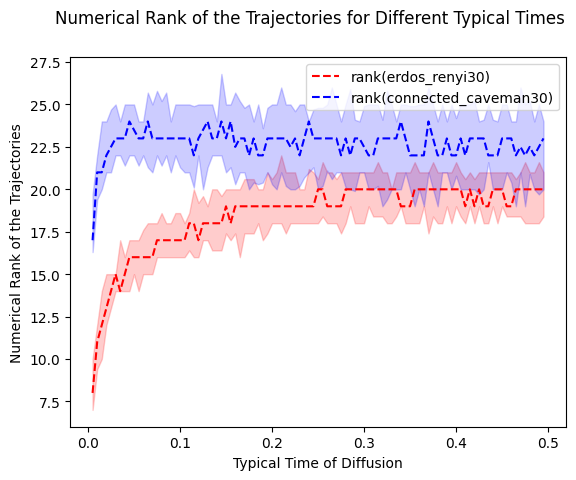}
    \caption{Numerical rank, of the observation matrix, for a fixed family of sample times, and two random graphs of $30$ nodes, for various diffusion rates $\tau$.}
    \label{fig:rank_diff_rate_infty}
\end{figure}
\subsection{Experiments: Diffusion Rate, Sampling Frequency}
\label{sec:rcn_diff_rate_sampling}
Let us now investigate the consequences of this phenomenon, for the practical reconstruction problem. We first show on \refig{full_obs_diffusion_rate_traj_errors} the AUC as a function of the typical time of diffusion $\tau$, for a fixed sampling rate. The AUC increases as the typical time of diffusion decreases, as we expected. It is bigger for the Relaxed Caveman graph, which has \gu{more structure} than the \erdosrenyi one.
\begin{figure}
    \centering
    \includegraphics[scale=.45]{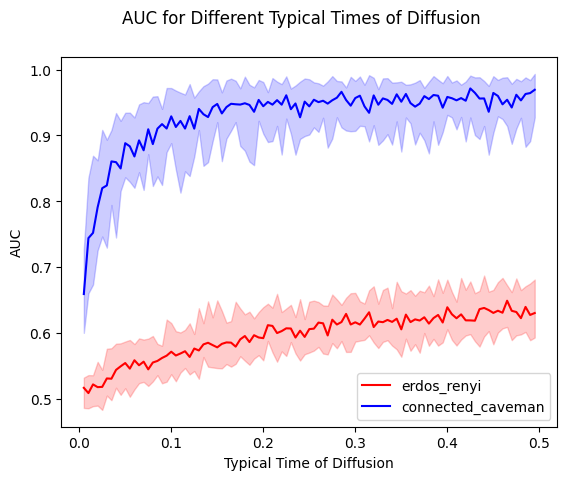}
    \caption{AUCs for Various Diffusion Rates}
    \label{fig:full_obs_diffusion_rate_traj_errors}
\end{figure}
Then, we study how increased sampling may help prediction for high diffusion rates. We therefore ran experiments for different values of $\tau$, and different sampling rates.
On \refig{diff_rate_sample_step_heatmap}, we show a heatmap of the AUC, with different sampling steps, and diffusion rates, for an \erdosrenyi graph of $30$ nodes. The darker the color, the smaller the AUC is. On each row, we see colors get darker as we go to the right: this means that, for each fixed diffusion rate, the AUC deteriorates as the sampling step increases. On each column, we see colors get darker as we move to the top: this means that, for each sampling step, the AUC worsens as the diffusion rate increases. Overall, we see that the bottom left triangle is lighter (sampling is high enough with respect to the diffusion rate, AUCs are big), while the top right triangle is darker (sampling is low with respect to the diffusion rate, AUCs are lower).

\begin{figure}
    \centering
    \includegraphics[scale=.45]{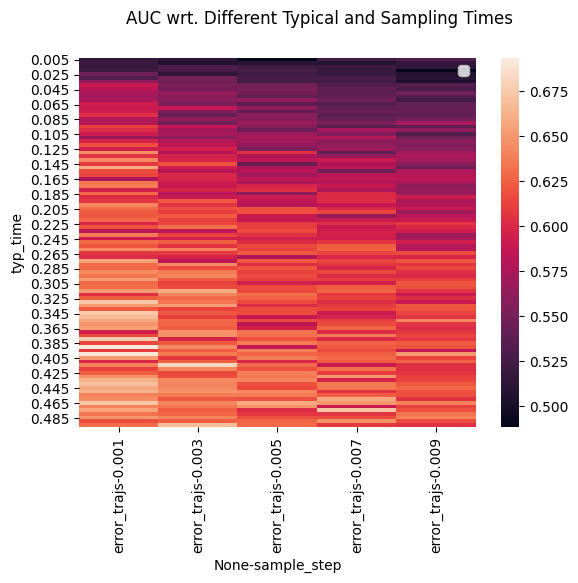}
    \caption{Heatmap of AUCs wrt. Typical Time of Diffusion and Sampling Step}
    \label{fig:diff_rate_sample_step_heatmap}
\end{figure}

\section{Influence of the Network Topology}
\label{sec:inf_topology}
We now study the influence of the network topology, first theoretically, from an algebraic standpoint (\resec{symmetries}), then experimentally (\resec{rcn_topology}).
\subsection{Symmetries}
\label{sec:symmetries}

Thanks to \relem{set_same_traj} we know that identifiability of $\mdpl$ is linked to the dimension of the vector space spanned by the flow $\flot(\modl)$. Now, symmetries of the model may cause the trajectories to live in low dimensional spaces. Indeed, they often lower the dimensions of the studied spaces by eliminating the dependencies of equations in some variables. For instance, a 2 dimensional problem in physics which is invariant under rotations around the origin will have a solution which will only depend on the distance to the origin. These principles have been applied successfully to numerous fields, and have been used in the context of mathematical epidemiology \cite{simon2011, ward2019} to reduce the number of calculations needed to simulate the propagation of diseases. We now investigate the influence of symmetries on the inverse problem we study.

Let us first define precisely symmetries. We write $\mathfrak{S}_N$ the symmetric group of order $N$, and $\sigma$ its elements, which are called permutations. We write $P(\sigma)$ the permutation matrix associated to the permutation $\sigma$. Then, for any $V\in\mathbb{R}^N$, $P(\sigma)V=V$ if, and only if, for every orbit of $\sigma$, for every $i,j$ in this orbit, we have $V_i=V_j$. A vector $X=\tpl{S}{I}{R}$ is symmetric with respect to $\sigma$ if $P(\sigma)S=S$, and likewise for $I$ and $R$. This notion extends to groups of permutations, as follows.
Let $\sglin$ be a subgroup of $\mathfrak{S}_N$, and define $\fix(\sglin)$ \cite{lang2012algebra} the space of vectors stable by $\sglin$, that is:
\begin{equation*}
\fix(\sglin) = \enstq{V \in \mathbb{R}^N}{\forall \sigma \in \sglin, P(\sigma)V = V}.    
\end{equation*}
Then, $X$ is said to be symmetric with respect to $\sglin$ if $X\in\fix(\sglin)$.
This extends to flows by saying that a flow $\tpl{S}{I}{R}=\flot(\modl)$ is symmetric with respect to $\sglin$ if $\flot(\modl)\in\fix(\sglin)$.
Finally, we say that a model $\modl=\tpl{\mdpl}{\cpl{\vecbeta}{\vecdelta}}{X_0}$ is symmetric with respect to some permutation $\sigma$ if, writing $P=P(\sigma)$, we have $\mdpl=P\mdpl P^{-1}$, $P\vecbeta=\vecbeta$, $P\vecdelta=\vecdelta$ and $PX_0=X_0$.  We then say $\sigma$ is an automorphism of $\modl$, extending in a straigthforward way the notion of graph automorphism \cite{hell2004graphs}. 
Indeed, if $\sigma$ is an automorphism of $\modl$, then it is in particular an automorphism of the underlying weighted graph, meaning that for all nodes $i, j \in \nods$, the edges $i\leadsto j$ and $\sigma(i)\leadsto\sigma(j)$ have the same weight: $\mdpl_{i, j} = \mdpl_{\sigma(i), \sigma(j)}$.
We write $\aut(\modl)$ the group of model automorphisms of $\modl$.

We first establish, in \relem{classification}, that trajectories generated by a diffusion $\mdpl$ are symmetrical with respect to some group $\sglin$ first if, and only if, there exists a diffusion admitting all permutations in $\sglin$ as automorphisms which generates the same trajectories and secondly if, and only if, the diffusion $\mdpl$ stabilizes $\fix(\sglin)$.
\begin{lemma}[Networks Generating Symmetrical Trajectories]
\label{lem:classification}
Let $\mdpl$ be a diffusion matrix, $\vecbeta, \vecdelta$ be the vectors of epidemiological coefficients, and $\sglin$ be a subgroup of $\mathfrak{S}_N$. 
Assume that $\vecbeta$ and $\vecdelta$ are symmetric with respect to $\sglin$. Then, the following conditions are equivalent.
\begin{enumerate}
    \item Symmetries of the Trajectories. For all $X_0\in\fix(\sglin)$, the flow of $\tpl{\mdpl}{\cpl{\vecbeta}{\vecdelta}}{X_0}$ is symmetric with respect to $\sglin$. 
    \item Symmetrical Generating Diffusion. There exists a diffusion matrix $\overline{\mdpl}$ such that $\sglin \subset \aut\tpl{\overline\mdpl}{(\vecbeta, \vecdelta)}{X_0}$ and such that for all $X_0\in\fix(\sglin)$, the flow of $\tpl{\overline\mdpl}{(\vecbeta, \vecdelta)}{X_0}$ equals the flow of $\tpl{\mdpl}{(\vecbeta, \vecdelta)}{X_0}$. 
    \item Stabilization by the Diffusion. $\mdpl$ stabilizes $\fix(\sglin)$, that is $\mdpl\fix(\sglin) \subset \fix(\sglin)$. 
\end{enumerate}
\end{lemma}

We now show that, in the spirit of \relem{set_same_traj}, trajectories symmetrical with respect to $\sglin$ are generated by diffusions which differ by a matrix $Z$ vanishing on $\fix(\sglin)$. These matrices represent the fact flows between nodes with identical $S$, $I$ and $R$ values may be redirected freely within themselves, provided the outgoing flows are modified accordingly\footnote{In fact, $Z$ matrices, like the diffusion matrices, describe rates. However, as long as nodes have equal values, modifying the rates, or the flows going out of them, becomes equivalent.}. The nodes where the flows are identical are those in the same orbits under $\sglin$ \cite{lang2012algebra}, that is the nodes $i$ and $j$ such that, for some $\sigma\in\sglin$, we have $j=\sigma(i)$.
Define therefore, for all $1\leq i < j \leq N$, and for all $1 \leq k \leq N-1$, the redirection matrix
\begin{equation*}
Z^{i,j,k}=E_{k,i} - E_{k,j} - E_{N,i} + E_{N,j},
\end{equation*}
where the $E_{r,s}$'s matrices are the vectors of the canonical basis of $\mathcal{M}_N\paren{\mathbb{R}}$. This matrix removes one unit of rate from the edge $j\leadsto k$, and adds one unit of rate on the edge $i\leadsto k$. It does the reverse with respect to the node $N$, taking one unit of rate from $i\leadsto N$ and adding it to $j\leadsto N$, in order to enforce the fact that the sums of $Z^{i,j,k}$ vanish, that is as much rate goes to each node than goes out.
\begin{lemma}[Flow Redirection within the Orbits]
\label{cor:synthesis}

Under the same assumptions as in \relem{classification}, let $\sglin$ be the biggest group of symmetries letting invariant the trajectories.
Then, the affine space of matrices producing the same trajectories as $\mdpl$ for every initial condition $X_0 \in \fix(\sglin)$ is exactly the subspace generated by the $Z^{i, j, k}$'s, for all $i$ and $j$ which are in the same orbit under $\sglin$. This space has dimension at least
\begin{equation*}
(N-1) (N - \# \{\text{different trajectories}\}).    
\end{equation*}
\end{lemma}
The dimension of this space is a lower bound on the dimension of the affine space of matrices generating the same trajectories as $\mdpl$.
To summarize, given a diffusion matrix $\mdpl$, we have given an explicit description of a set of matrices giving the same trajectories as $\mdpl$. As a result, if the diffusion matrix we try to reconstruct gives symmetrical trajectories, and if we have an algorithm which gives us one solution of the reconstruction problem, then we are able to find many such matrices explicitly, though we cannot single the original $\mdpl$ out.
Note that this has consequences on the conditioning of the observation matrix. Indeed, its rank is then necessarily bounded by $\# \{\text{different  trajectories}\}$. As such, if the model presents symmetries, then several singular values of the observation matrix will be zero, and in a neighbourhood of $\mdpl$ as well, the numerical rank will be bounded by $N - \# \{\text{different  trajectories}\}$. This proves that the nearest a model is to a symmetrical model, the most difficult it is to reconstruct the diffusion matrix.

\subsection{Network Topology Experiments}
\label{sec:rcn_topology}
We now study experimentally the influence of the network topology on the estimation and prediction problems.
As explained in \resec{exp_setup}, we study two metrics: the AUC on the presence or absence of edges, and the prediction error, and we present results for various sizes of graphs, and various types of random networks, exhibiting different topologies. The plots are box plots, where the solid lines are the medians of values, and the shaded areas gather the [10\%, 90\%] intervals of values.

We present, on \refig{full_obs_auc}, the AUC as a function of the number of nodes, for various types of random graphs. The full set of parameters used is available in the code. The AUC is quite good for small graphs, more than $0.8$, but decreases as the number of nodes increases. As expected, the more constrained the topology, the better the AUC. Indeed, it is in general best for the \brbalbert graph, and the second best is often the Relaxed Caveman graph. The Waxman gaph, and above all the \erdosrenyi one, exhibit the worse AUCs.  

On \refig{full_obs_traj_errors}, we show the prediction error.
We see the prediction errors in general are quite low, less than $2\times 10^{-4}$, and diminish with the number of nodes. Moreover, the \erdosrenyi graph consistently exhibits the lowest error. These results are consistent with each other, in the sense that it seems the more the graph has connections, the easiest it is to predict the future behaviour of the system (more edges, either through more nodes, or through the topology, in the case of the \erdosrenyi graph). However, they are opposite to the results for the AUCs. Therefore, they tend to suggest that the more constrained the topology, the easier it is to reconstruct the network, but the more mixing there is, the easiest it is to predict the future evolution of the system.
We did not display the prediction errors for the \brbalbert graph,  as it was about 5 times greater than for the other graphs, and exhibited also high variance. We believe it comes probably first from the fact it is proner to numerical instabilities, as we said in \resec{exp_setup}. Secondly, it is also probably due to its topology being more constrained: as a result, small errors on the reconstruction lead to high errors on the prediction.

Finally, on \refig{full_obs_num_ranks}, we see the numerical rank of the observation matrix for several graphs. It tends to stagnate or decrease as the number of nodes increases, which is not surprising, as large matrices tend to have small singular values, which therefore do not contribute to the numerical rank. It is consistently higher for the \brbalbert graph, which structure is more constrained.
\begin{figure}
    \centering
    \includegraphics[scale=.45]{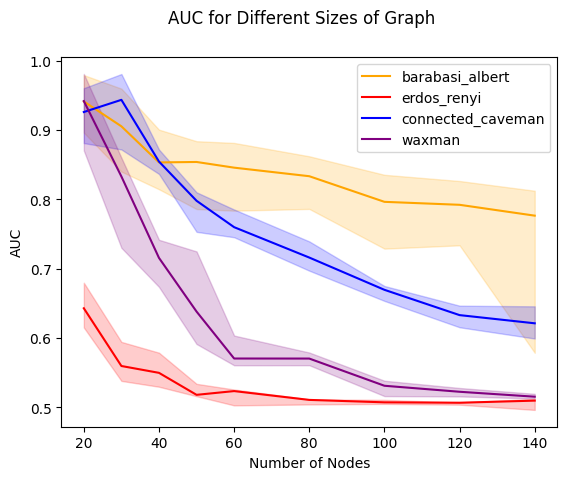}
    \caption{AUC on the Adjacency Matrix for Various Graphs and Sizes of Graphs}
    \label{fig:full_obs_auc}
\end{figure}
\begin{figure}
    \centering
    \includegraphics[scale=.45]{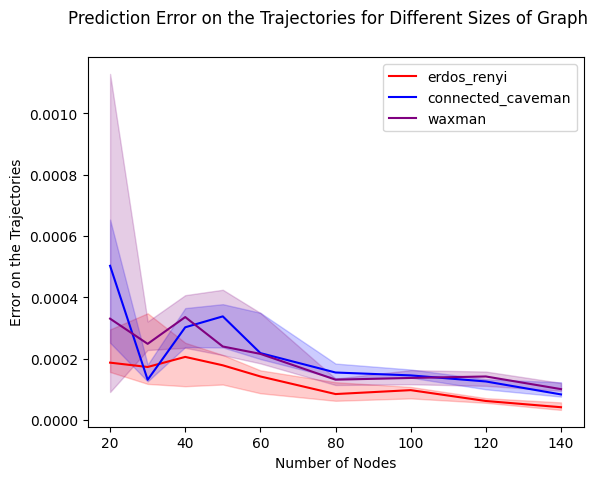}
    \caption{Prediction Error on the Trajectories for Various Sizes of Graphs}
    \label{fig:full_obs_traj_errors}
\end{figure}
\begin{figure}
    \centering
    \includegraphics[scale=.45]{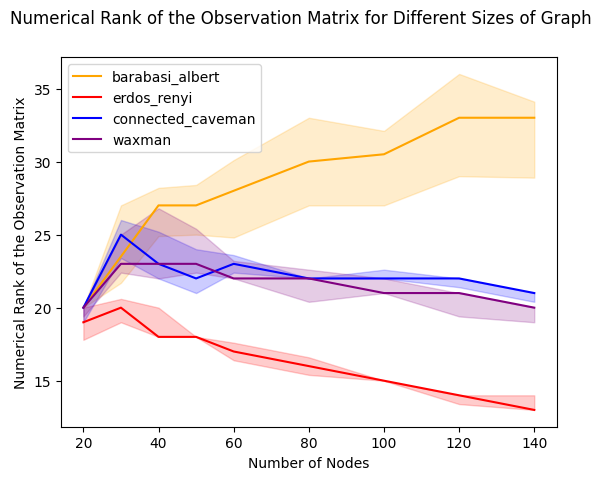}
    \caption{Numerical Ranks of the Observation Matrix for Various Sizes of Graphs}
    \label{fig:full_obs_num_ranks}
\end{figure}

\section{Conclusions, future works}
\label{sec:ccl_future_works}
In this article, we have studied the reconstruction, and prediction, problems, for an epidemic reaction-diffusion. We have proved that for almost every network, the reconstruction problem is identifiable. Then, we have shown that the quicker the diffusion, the lower the numerical rank of the observation matrix, and the harder the reconstruction, but that increasing sampling helped reconstruct the network.
Then, we have classified symmetrical networks generating the same trajectories. Finally, we showed experimentally, on synthetic data constructed with random graph generators exhibiting different topologies, that reconstruction was easier for more \gu{constrained} topologies, and that the prediction problem could still be solved satisfyingly even when the network topology makes exact reconstruction difficult.

We have studied the case when the observations we have consist of all the trajectories in all the nodes. Other studies have considered partial observations, as in the presence of missing nodes \cite{tyrcha2014, haehne2019}, or more generally partial observations \cite{nitzan2017, ioannidis2018, ioannidis2019}. This would be an interesting extension to further our work. 

Then, measures often come with a level of noise, and studying the robustness of reconstruction in the presence of noise should be another direction of study.

Finally, studies on real-world data would allow to attempt at separating the specific influence of diffusion, through transportation networks for instance, on the spread of epidemics, from that of internal (within cities, or countries) reactions.

\printbibliography

\newpage
\appendix

\section{Proofs for \resec{identifiability}, Identifiability of the Diffusion}
\label{app:model_problems}
\begin{proof}[\relem{set_same_traj}]
Let $\mdpl=\mdpl^*+H$ be a diffusion matrix, such that for all $t\geq 0$, we have $S(t)\in\ker(H)$, $I(t)\in\ker(H)$ and $R(t)\in\ker(H)$. Let us check $\tpl{S}{I}{R}=\flot(\mdpl^*)$ is then a solution of \reeq{sir-metapop} with diffusion matrix equal to $\mdpl$. Indeed, for all $t\geq 0$, we have
\begin{equation*}
\ba
\frac{dS}{dt}(t) &= -\vecbeta S(t) \odot I(t) + \mdpl^* S(t) \quad \text{by assumption}\\
&= -\vecbeta S(t) \odot I(t) + \mdpl^* S(t) + H S(t) \quad \text{as} \, S(t)\in\ker H \\
&= -\vecbeta S(t) \odot I(t) + \mdpl S(t).
\ea
\end{equation*}
Likewise, $I$ and $R$ satisfy the corresponding equations, with diffusion matrix $\mdpl$ instead of $\mdpl^*$. Now, $\tpl{S}{I}{R}$ satisfies the differential equation \reeq{sir-metapop} with diffusion matrix $\mdpl$, and starts at the initial condition $X_0$. Therefore, by unicity of the solutions of this system, we have $\tpl{S}{I}{R}=\flot(\mdpl)$.
\end{proof}

To prove \relem{traj_oft_whole_space}, we first prove that the trajectories generated by a \emph{linear} system often generate the whole space. Before doing this, we need the following technical lemmas.
\begin{lemma}[Technical Result]
\label{lem:tech_lem}
For every integer $N\geq 1$, let $\paren{\lambda_\node}_{1 \leq \node \leq N}$ and $\paren{t_k}_{1 \leq k \leq N}$ be two families of distinct (within each family) real numbers. Then, for every $N \geq 1$, the following property holds true:
\begin{multline*}
\mathcal{P}(N): \quad \forall 1 \leq \node \leq N, \sum_{k=1}^N \mu_k e^{\lambda_\node t_k} = 0
\Rightarrow (\forall 1 \leq k \leq N, \mu_k = 0).
\end{multline*}
\end{lemma}

\begin{proof}
We prove it by induction on $N$. The case $N = 1$ is immediate. Let $N \geq 2$, and let us show that if $\mathcal{P}(N-1)$ holds, then $\mathcal{P}(N)$ holds as well.
Assume that for all $1\leq \node \leq N$, we have $\sum_{k=1}^N \mu_k e^{\lambda_\node t_k} = 0$. Then, for all $1 \leq \node \leq N$, we have
\begin{equation*}
\sum_{k=1}^{N-1} \mu_k e^{\lambda_\node (t_k- t_N)} + \mu_N = 0.
\end{equation*}
Let us consider the mapping:
\begin{equation*}
g : \mathbb{R}\ni\lambda \mapsto \sum_{k=1}^{N-1} \mu_k e^{\lambda (t_k- t_N)} + \mu_N.
\end{equation*}
Then, $g$ has $N$ distinct roots (the $\lambda_\node$'s), therefore by Rolle's theorem, its derivative admits $N-1$ distinct roots.
We therefore obtain $N-1$ values $\lambda_\node'$ such that, for all $1\leq \node \leq N-1$, $\lambda_\node < \lambda_\node' < \lambda_{\node+1}$ and $\frac{dg}{d\lambda}(\lambda_\node')=0$.
As a result, for all $1\leq \node \leq N-1$, we have
\begin{equation*}
0 = \frac{dg}{d\lambda}(\lambda'_\node) = \sum_{k=1}^{N-1} \mu_k (t_k - t_N) e^{\lambda_\node (t_k- t_N)}.
\end{equation*}
By induction, for all $1 \leq k \leq N - 1$, we have $\mu_k (t_k - t_N) = 0$, and therefore $\mu_k = 0$, as the $t_k$'s are distinct. Then, $\mu_N = 0$ as well, hence the result.
\end{proof}
We need the following lemma. Though well-known, we could not locate a reference, so added it here for completeness.
\begin{lemma}[Almost Every Diffusion Matrices has Distinct Eigenvalues]
\label{lem:diffusion_matrix_distinct_eigenvalues}
Almost every diffusion matrix has $N$ distinct eigenvalues.
\end{lemma}
\begin{proof}
Let $\mathcal{E}$ be the subset of diffusion matrices which do not have $N$ distinct eigenvalues, and let us prove it has null measure. We know that $\mathcal{E}$ is the zero set of the discriminant $\Delta$ applied to the characteristic polynomial \cite{lang2012algebra}, defined for every matrix $A$ of order $N$ by
\begin{equation*}
F(A) = \Delta \circ \det(X \id - A).
\end{equation*}
The set of diffusion matrices is contained in the linear subspace of matrices $A$ satisfying $1_N^T A = 0$, which is of dimension $N^2-N$. We equip it with the standard Lebesgue measure over $\mathbb{R}^{N^2-N}$.

Then, for any diffusion matrix $D \not\in \mathcal{E}$, and any diffusion matrix $\mdpl$, $D + \lambda \mdpl$ is in $\mathcal{E}$ only for a finite number of $\lambda\in\mathbb{R}$.
Indeed, $\lambda \mapsto F(D+\lambda \mdpl) = \Delta \circ \det(X Id - D - \lambda \mdpl)$ is a polynomial function, so it is either zero or has a finite number of zeroes. But then it is nonzero at $\lambda = 0$, because $D \not\in \mathcal{E}$.
Therefore, the indicator function $\chi$ of $\mathcal{E}$ is zero almost everywhere on $\{D + \lambda \mdpl, \lambda \in \mathbb{R}\}$, so it is also zero almost everywhere on $\{\mdpl + \lambda D, \lambda \in \mathbb{R}\}$.

Let us finally fix some $D\notin\mathcal{E}$, and some vector space $\mathcal{H}$ such that $H \oplus \mathrm{Vect}(D) = \{A , 1_N^T A = 0\}$. Thanks to the Fubini-Tonelli theorem, we have
\begin{equation*}
\ba
\int_{\mathbb{R}^{N^2-N}} \chi(A) dA &= \int_\mathcal{H} \int_{\mathbb{R}} \chi(A + \lambda D) d\lambda dA\\
&= \int_\mathcal{H} 0 = 0.
\ea
\end{equation*}
Therefore, $\mathcal{E}$ has zero Lebesgue measure.
\end{proof}

\begin{lemma}[Almost Everywhere Identifiability, linear case]
\label{lem:trajs_whole_space_linear}
Let $0 \leq t_1 < ... < t_N < \infty$ be a subdivision of the nonnegative real half-axis, $\mdpl \in \mathcal{M}_N(\mathbb{R})$, and $y_0 \in \mathbb{R}^N$. Let $y$ be the solution of the following differential equation:
\begin{equation}
\left\lbrace \ba
\frac{dy}{dt} &= \mdpl y \\
y(0) &= y_0.
\ea \right.
\end{equation}
Then, $(y(t_1), ..., y(t_N))$ is a basis of $\mathbb{R}^N$ for almost every $\mdpl, y_0$.
\end{lemma}
\begin{proof}
Let us prove the result for $\mdpl$ and $y_0$ satisfying the additional assumptions that all the eigenvalues of $\mdpl$ have multiplicity $1$, and that every coordinate of $y_0$ in an eigenbasis of $\mdpl$ is nonzero (note that as $\mdpl$ has $N$ distinct eigenvalues, the associated subspaces are 1-dimensional, so the eigenbasis is unique up to permutation or scaling of the vectors). Thanks to \relem{diffusion_matrix_distinct_eigenvalues}, we will then have proved the result as stated, that is for almost every $\mdpl$, and also almost every $y_0$ (since the $y_0$'s with at least one 0 coordinate live in a union of $N$ hyperplanes $x_i = 0$ where $x_i$ is i-th the coordinate in the eigenbasis, which has zero Lebesgue measure.).

Since $y$ satisfies the linear equation $dy/dt = \mdpl y$, we know that, for all $t\geq 0$, we have $y(t) = \exp(t \mdpl) S_0$.
Let us write $\lambda_1, ..., \lambda_N$ the distinct eigenvalues of $\mdpl$, and $(e_{\node})$ a corresponding eigenbasis. For all $t\geq 0$, we can decompose $y(t)$ along this eigenbasis. Let us write $y_1(t),\ldots, y_N(t)$ the corresponding coefficients so that, for all $t\geq 0$, we have $S(t) = \sum_\node y_{\node}(t) e_{\node}$. By assumption, for every $\node$, we have $y_{\node}(0) \neq 0$. As a result, for every $t \geq 0$, we have
\begin{equation*}
\ba
    y(t) &= \exp\paren{\mdpl t} y_0
    = \exp\paren{\mdpl t} \sum_\node y_\node(0) e_\node \\
    &= \sum_\node \exp(tM) y_\node(0) e_\node
    = \sum_\node \exp(t\lambda_\node) y_\node(0) e_\node.
\ea
\end{equation*}
We want to show that $(y(t_k))_{1 \leq k \leq N}$ is a basis of $\mathbb{R}^N$. Let $\sum_k \mu_k y(t_k) = 0$ be a linear dependence relation. Since we have
\begin{equation*}
\ba
    \sum_{k=1}^N \mu_k y(t_k) &=
    \sum_{k=1}^N \mu_k \sum_\node e^{t_k \lambda_\node} y_\node(0) e_\node \\
    &= \sum_\node y_\node(0) e_\node \sum_{k=1}^N \mu_k e^{t_k\lambda_\node},
\ea
\end{equation*}
we know that, for all $\node \in \nods$, we have
\begin{equation*}
\sum_{k=1}^N \mu_k e^{t_k \lambda_\node}=0,
\end{equation*}
using the unicity of coordinates in the basis $(e_\node)_\node$ and the fact that for all node $\node$, we have $y_\node(0) \neq 0$. We use \relem{tech_lem} to conclude.
\end{proof}

We can now prove \relem{traj_oft_whole_space}.
\begin{proof} Let $K = S + I + R$ be the total population irrespective of infection status (for each node $\node$, $K_\node=S_\node+I_\node+R_\node$ is the population of node $\node$). By definition, for all $t \geq 0$, $K(t)$ is in the space generated by the trajectories. Then, $K$ follows the differential equation
\begin{equation*}
\ba
\frac{dK}{dt} 
&= \frac{dS}{dt} + \frac{dI}{dt} + \frac{dR}{dt} \\
&= - \beta S \odot I + \mdpl S + \beta S \odot I - \delta I + \mdpl I + \delta I + \mdpl R \\
&= \mdpl K.
\ea
\end{equation*}
Using \relem{trajs_whole_space_linear}, we see that $(K(t_1), ..., K(t_N))$ generates $\mathbb{R}^N$ for almost every $\mdpl$, and $K(0) = S_0 + I_0 + R_0$, so for almost every $\mdpl, S_0, I_0, R_0$. As a result, for almost every $\mdpl$, $X_0=\tpl{S_0}{I_0}{R_0}$, the space generated by the trajectories contains a family which generate $\mathbb{R}^N$. This proves our claim.
\end{proof}

\section{Proofs for \ref{sec:diff_rate_analysis}, analysis of the influence of the diffusion rate}

Let us first prove the following result.
\begin{lemma}[Trajectories Close to a Line for Infinitely Quick Diffusion]
\label{lem:traj_close_line_infty_quick_diffusion}
Let us assume the initial condition $X_0$ is such that $S_0$, $I_0$ and $R_0$ are proportional to the stationary distribution. 

For every $\tau > 0$, let us write $\flot^\tau$ the solution of the system of \reeq{sir-metapop} where $\mdpl$ is replaced by $\mdpl/\tau$, that is, for every $t\geq 0$, $\flot_t^\tau = \tpl{S_\tau(t)}{I_\tau(t)}{R_\tau(t)}$.
Then, for all $T > 0$,
\begin{equation*}
\sup_{0 \leq t \leq T} \dist{S_\tau(t)}{\mathbb{R}\statiod} \to 0,
\end{equation*}
when $\tau\to\infty$, and likewise for $I_\tau$ and $R_\tau$.
\end{lemma}
The results extends to cases when $X_0$ is not proportional to the stationary distribution, only taking the supremum over some interval $[t(\tau), T]$, where $t(\tau)$ tends to $0$, when $\tau\to 0$, and $t(\tau)$ represents the time it takes for the system to converge to the stationary distribution.
\begin{proof}
First, for all $\tau>0$, and $t\geq 0$, we have
\begin{equation}
\label{eq:nutau}
S_\tau(t) = \paren{\sum_\node S_{\tau,\node}(t)} \statiod + \nu_\tau(t),
\end{equation}
where $\nu_\tau(t)$ belongs to the set $\mathcal{H}=\enstq{\nu\in\mathbb{R}^N}{\sum_\node \nu_\node=0}$\footnote{This is a consequence of the decomposition $\mathbb{R}^N = \mathbb{R} \statiod \oplus \mathcal{H}$.}. Moreover, $\nu_\tau(t)$ is bounded uniformly in $\tau >0$ and $t\geq 0$, as all the $S_\tau(t)$'s are bounded by the total population, that is the sum of the coordinates of the initial condition $X_0$.

Then, for all $\tau >0$, $t\mapsto\nu_\tau(t)$ is differentiable thanks to \reeq{nutau} and, by differentiating \reeq{nutau}, we see $\nu_\tau$ satisfies a differential equation of the form:
\begin{equation*}
\frac{d \nu_\tau}{dt}(t) = \frac{\mdpl}{\tau}\nu_\tau(t) + \gamma_\tau(t),
\end{equation*}
where $\gamma_\tau$ is a quantity depending on many things, but which is uniformly bounded in $\tau>0$ and $t\geq 0$, again thanks to the fact that the $S_\tau(t)$'s, the $I_\tau(t)$'s and the $R_\tau(t)$'s are bounded by the total size of the population. Moreover, for all $\nu\in\mathbb{R}^N$, we have $\mdpl \nu\in\mathcal{H}$, since the columns of $\mdpl$ sum to $0$. As a result, for all $\tau >0$, and all $t\geq 0$, we have $\frac{d\nu_\tau}{dt}(t)\in\mathcal{H}$ (since $\mathcal{H}$ is a finite dimensional vector space, so derivatives of functions living on it stay in it), and $\frac{\mdpl}{\tau}\nu_\tau(t)\in\mathcal{H}$ (by what precedes), so that $\gamma_\tau(t)$ also belongs to $\mathcal{H}$.

Finally, for all $\tau>0$, and all $t\geq 0$, we have
\begin{equation*}
\ba
\nu_\tau(t) = &\exp\paren{\mdpl\frac{t}{\tau}}\nu_\tau(0)
+ \int_0^t \exp\paren{\mdpl\frac{t-s}{\tau}} \gamma_\tau(s) \, ds \\
&=\int_0^t \exp\paren{\mdpl\frac{t-s}{\tau}} \gamma_\tau(s) \, ds,
\ea
\end{equation*}
since $\nu_\tau(0)=0$, as the initial condition is proportional to $\statiod$ by assumption.
Let us now fix $\eps>0$.
Let $\bo\subset\mathcal{H}$ be a ball such that that, for all $\tau>0$ and $t\geq 0$, we have $\gamma_\tau(t)\in\bo$. Since $\mdpl$ only has eigenvalues with (strictly) negative eigenvalues on $\mathcal{H}$, we may find some threshold $u_\mathrm{min} >0$ such that, for all $u\geq u_\mathrm{min}$, for all $\nu\in\bo$, we have
\begin{equation*}
\nrm{\exp\paren{\mdpl u}\nu} < \eps.
\end{equation*}
Moreover, there exists a constant $\kappa\geq 1$ such that, for all $u\geq 0$, for all $\nu\in\bo$, we have $\nrm{\exp\paren{\mdpl u}\nu} \leq \kappa$.
Let us now consider $\tau \leq \inv{u_\mathrm{min}}\frac{\kappa}{\eps}$ (it is chosen so that, for $t-s\geq \frac{\eps}{\kappa}$, we have $\frac{t-s}{\tau} \geq u_\mathrm{min}$).
Therefore, for all $t \leq T$, we have
\begin{equation*}
\ba
\nrm{\nu_\tau(t)} &\leq
\int_0^t \nrm{\exp\paren{\mdpl\frac{t-s}{\tau}} \gamma_\tau(s)} \, ds \\
&= \int_0^{\frac{\eps}{\kappa}}  \nrm{\exp\paren{\mdpl\frac{t-s}{\tau}} \gamma_\tau(s)} \, ds \\
&+ \int_{\frac{\eps}{\kappa}}^T  \nrm{\exp\paren{\mdpl\frac{t-s}{\tau}} \gamma_\tau(s)} \, ds \\
&\leq \frac{\eps}{\kappa} \kappa + \int_{\frac{\eps}{\kappa}}^T \eps \, ds \\
&\leq \eps + \eps \paren{T-\frac{\eps}{\kappa}} \leq \eps \, \paren{1+T}.
\ea
\end{equation*}
As a result, for all $\tau\leq \inv{u_\mathrm{min}}\frac{\kappa}{\eps}$, with $u_\mathrm{min}$ and $\kappa$ chosen independently of $\tau$, we have $\sup_{0\leq t \leq T} \nrm{\nu_\tau(t)} \leq \eps\paren{1+T}$. We have therefore proven that $\dist{S_\tau(t)}{\mathbb{R}\statiod}\to 0$, as $\tau\to 0$. We would prove likewise the result for $I_\tau$ and $R_\tau$, which concludes the proof.
\end{proof}

We can now prove \relem{limit_trajs_diffusion_rate_infty}.
\begin{proof}
Let $T >0$.
Now, let us write, for all $t\geq 0$, using the notations of the proof of \relem{traj_close_line_infty_quick_diffusion},
\begin{equation}
\label{eq:decomp_vect_suscep}
S^\tau(t) = s_\tau(t)\statiod + \nu_\tau(t),
\end{equation}
and likewise for $I^\tau$ and $R^\tau$. Since $S_0$, $I_0$ and $R_0$ are proportional to the stationary distribution, we know, thanks to \relem{traj_close_line_infty_quick_diffusion}, that $\nu^\tau(t)$ tends towards $0$, uniformly on each $[0,T]$, with $T \geq 0$.
Since for all $\tau$, $S^\tau$, $I^\tau$ and $R^\tau$ are nonnegative, and bounded by the total population, the family of functions $\tpl{s^\tau}{i^\tau}{r^\tau}$, defined on $[0,T]$, has values in a bounded set of the continuous functions from $\mathbb{R}_+$ to $\mathbb{R}^{3}$, endowed with the infinity norm on each of $s$, $i$ and $r$, that is $\nrmi{\tpl{s}{i}{r}}=\max\tpl{\nrmi{s}}{\nrmi{i}}{\nrmi{r}}$. Moreover since, for all $\tau \geq 0$, for all $0 \leq t \leq T$, we have
\begin{equation*}
S^\tau(t) = S^\tau(0) - \int_0^t \vecbeta S^\tau(s) \odot I^\tau(s) + \frac{\mdpl}{\tau} S^\tau(s) \, ds,
\end{equation*}
we obtain, summing along the coordinates,
\begin{equation*}
s^\tau(t) = s^\tau(0) - \sum_\node \beta_\node \statiod(\node)^2 \int_0^t s^\tau(s) i^\tau(s) \, ds.
\end{equation*}
As a result, $s^\tau$ is differential on $[0,T]$, its derivative satisfies
\begin{equation*}
\frac{ds^\tau}{dt} = - \paren{\sum_\node \beta_\node \statiod(\node)^2} s^\tau(t)i^\tau(t),
\end{equation*}
and its derivative is threfore bounded on $[0,T]$, uniformly on $\tau$. The same holds for $i^\tau$ and $r^\tau$. Therefore, $\tpl{s^\tau}{i^\tau}{r^\tau}$ is equi-continuous. As a result, the family $\paren{\tpl{s^\tau}{i^\tau}{r^\tau}}_{\tau \geq 0}$ is pre-compact \cite{suth2004} in the Banach space of functions from $[0,T]$ to $\mathbb{R}^3$, endowed with the infinity norm defined above, so that, provided it admits an unique adherence value, it converges towards this one.

Let us consider a converging subsequence, and still index it by $\tau$, to simplify notations. 
As a result, the limit $s$ satisfies
\begin{equation*}
s(t) = s(0) + \sum_\node \beta_\node \statiod(\node)^2 \int_0^s s(s) i(s) \, ds,
\end{equation*}
and likewise for $i$ and $s$. Moreover, for all $t\geq 0$, we know that $s^\tau(0)=\sum_\node S^\tau(0)=\sum_\node S_\node(0)$ which does not depend on $\tau$, therefore $s(0)=\sum_\node S_\node(0)$, and likewise for $i_0$ and $r_0$. Therefore, $\tpl{s}{i}{r}$ is solution of the scalar system described in the statement of the Lemma, and by uniqueness of the solutions of this system, satisfying the initial condition $\tpl{s_0}{i_0}{r_0}$, the tuple is uniquely defined. Therefore, the family of $\tpl{s^\tau}{i^\tau}{r^\tau}$'s admits a unique adherence value, and converges towards this one. Plugging back into \reeq{decomp_vect_suscep}, we see that $S^\tau \to s\statiod$, as $\tau\to 0$, uniformy on $[0,T]$, and likewise for $I^\tau$ and $R^\tau$, which concludes the proof.
\end{proof}

Let us prove \recor{rank_diff_rate_infty}.
\begin{proof}
Applying the results of \relem{limit_trajs_diffusion_rate_infty} with $T=t_K$, we know that the observation matrix writes
\begin{multline}
\label{eq:dvpt_obs_mat}
\obsmat\cpl{\paren{t_k}}{\flot^\tau(\mdpl)} = \\
\statiod \otimes \paren{s(t_1), \ldots, s(t_K), i(t_1), \ldots i(t_K), r(t_1)\,\ldots r(t_K)} \\
+ \go{\eps(\tau)},
\end{multline}
as $\tau\to 0$, 
where $\eps(\tau)\to 0$, when $\tau\to 0$.
Indeed, using the notations of the proof of \relem{limit_trajs_diffusion_rate_infty}, we know that for each $1\leq k \leq K$, the column $S(t_k)$ of the observation matrix (for instance), writes $S(t_k) = s(t_k)\statiod + \nu_\tau(t_k)$, and $\nu_\tau$ tends to $0$, as $\tau\to 0$, uniformly on $[0,t_K]$. The first term of \reeq{dvpt_obs_mat} is of rank $1$, as $s+i+r=1$, identically. The conclusion follows from the continuity of the numerical rank of a matrix (the numerical rank is the number of singular values greater than some threshold, and these values depend continuously on the matrix).
\end{proof}

\section{Proofs for \resec{symmetries}, symmetries}
\label{app:symmetries}
Let us first show the effect of node-renumbering on the trajectories.
\begin{lemma}[Node re-numbering]
\label{lem:perm_and_trajs}
Let $\mathcal{M}=\tpl{\mdpl}{\cpl{\beta}{\delta}}{X_0}$ be a model and $P$ a permutation matrix. Then, for all $t\geq 0$, we have 
\begin{equation*}
\ba
\Phi_t(P \cdot \mathcal{M}) &= P \cdot \flot_t(\mathcal{M}) \\
&= \tpl{PS(t)}{PI(t)}{PR(t)},
\ea
\end{equation*}
writing $\tpl{S(t)}{I(t)}{R(t)}=\flot_t(\modl)$, and using the conventions of \resec{notations}.
\end{lemma}
This implies immediately that if $P$ is the matrix of an automorphism of our model, then for all $t\geq 0$, we have $\Phi_t(P \cdot \mathcal{M}) = \Phi_t(\mathcal{M})$, and therefore $P\cdot\flot_t(\mdpl)=\flot_t(\mdpl)$. In other words, if $i$ and $j$ are in a same orbit of $P$, then the trajectories at nodes $i$ and $j$ are the same: for all $t\geq 0$, $S_i(t)=S_j(t)$, and likewise for $I$ and $R$.
\begin{proof}
Let us show that $(PS, PI, PR)$ is a solution of the differential equation also satisfied by the flow $\flot\tpl{P\mdpl P^{-1}}{\cpl{P\beta}{P\delta}}{PX_0}$, which is enough to conclude by unicity of the solutions sharing the same initial condition. Let $1 \leq i \leq N$. Then, we have
\begin{align*}
    \frac{d}{dt} [PS]_i &=  \left [P \frac{dS}{dt} \right ]_i\\
    &= \frac{dS_{\sigma(i)}}{dt} \\
    &= - \beta_{\sigma(i)} S_{\sigma(i)} I_{\sigma(i)} + [M S]_{\sigma(i)} \\
    &= - [P\beta]_{i} [PS]_{i} [PI]_{i} + [P M S]_{i} \\
    &= - [P\beta]_{i} [PS]_{i} [PI]_{i} + [(P M P^{-1}) (P S)]_{i}.
\end{align*}
Therefore, we have
\begin{equation*}
\frac{d PS}{dt} = - \paren{P \beta} \paren{\paren{PS}\odot \paren{PI}} + \paren{P\mdpl P^{-1}}\paren{PS}.
\end{equation*}
So $PS$ satisfies the required equation, and $I$, $R$ do as well, which we show using the same method, and which allows us to conclude.
\end{proof}
We can now prove \relem{classification}. 
\begin{proof}
Let us first prove $3) \Rightarrow 2)$.
Let $\sigma\in\sglin$. For any $x\in\sglin$, we have
 \begin{equation*}
 \ba
P(\sigma) \mdpl P(\sigma)^{-1}x&=
P(\sigma) \mdpl x \quad \text{as} \, x \in \fix(\sglin) \\
&= \mdpl x \quad \text{as} \, \mdpl x \in \fix(\sglin).
\ea
\end{equation*}
We now average the nodes which give the same trajectories, which is the standard method of averaging under a group action. Let
\begin{equation*}
\bar\mdpl = \frac{1}{\# \sglin} \sum_{\sigma \in \sglin} P(\sigma) \mdpl P(\sigma)^{-1}.
\end{equation*}
We therefore obtain by construction that, for all $X_0\in\fix(\sglin)$, we have $\sglin \subset \aut(\bar\mdpl, (\beta, \delta), X_0)$. The fact $\bar\mdpl$ and $\mdpl$ agree on $\fix(\sglin)$ is a direct consequence of the averaging.
Now, since $\bar\mdpl$ is symmetric with respect to $\sglin$, thanks to \relem{perm_and_trajs}, we know that the trajectories it generates are also symmetric with respect to $\sglin$. As a result, they belong to $\fix(\sglin)$. Therefore, they also satisfy the differential equation with $\mdpl$, as we have just proven that $\mdpl$ and $\bar\mdpl$ agree on $\fix(\sglin)$.

Then, $2) \Rightarrow 1)$ is a direct consequence of \relem{perm_and_trajs}.

To prove $1) \Rightarrow 3)$, let $S_0\in\fix(\sglin)$, and let us show that $\mdpl S_0\in\fix(\sglin)$. Choose $I_0$ and $R_0$ such that $I_0\in\fix(\sglin)$, and define $X_0=\tpl{S_0}{I_0}{R_0}$. For every $\sigma\in\sglin$, for every node $i$, we have
\begin{equation*}
\left[\mdpl S(0)\right]_{\sigma(i)} = \frac{dS_{\sigma(i)}}{dt}(0) + \beta_{\sigma(i)} S_{\sigma(i)}(0) I_{\sigma(i)}(0).
\end{equation*}
Now, $\beta_\sigma(i)=\beta_i$ by assumption on the coefficients, and $S_{\sigma(i)}(0) I_{\sigma(i)}(0)=S_i(0)I_i(0)$ by assumption. Moreover, we have
\begin{equation*}
\ba
\frac{dS_{\sigma(i)}}{dt}(0) &=
\lim_{t\to 0} \frac{S_{\sigma(i)}(t) - S_{\sigma_i}(0)}{t}
= \lim_{t\to 0} \frac{S_i(t) - S_i(0)}{t} \\
&= \frac{dS_i}{dt}(0),
\ea
\end{equation*}
where the second equality is a consequence of the fact that trajectories remain in $\fix(\sglin)$. As a result, we have
\begin{equation*}
\left[\mdpl S(0)\right]_{\sigma(i)}
= \frac{dS_i}{dt}(0) + \beta_i S_i(0) I_i(0)
= [\mdpl S(0)]_i.
\end{equation*}
Therefore, for every $\sigma\in\sglin$, we have $P(\sigma)\mdpl S_0=\mdpl S_0$ so that, by definition, we have $\mdpl S_0\in\fix(\sglin)$.

\end{proof}

\begin{proof}[\recor{synthesis}]
Let us first prove that $Z = \mdpl - \overline\mdpl$ vanishes on $\fix(\sglin)$. For all $X_0 \in \fix(\sglin)$, since $\mdpl$ and $\overline\mdpl$ produce the same trajectories, we have, for all $t\geq 0$,
 \begin{equation*}
\frac{dS}{dt} = - \beta S(t) \odot I(t) + \mdpl S(t) = - \beta S(t) \odot I(t) + \mdpl'S(t),
\end{equation*}
so that $ZS(t)=\paren{\mdpl-\mdpl'}S(t)=0$. As a result, $ZS(0)=0$. This is true for all $S(0) \in \fix(\sglin)$ (as $X_0=\tpl{S_0}{I_0}{R_0}$ is arbitrary provided $S_0, I_0$ and $R_0$ all belong to $\fix(\sglin)$), so $Z$ vanishes on ${\fix(\sglin)}$.
 
\end{proof}

\begin{proof}[\recor{synthesis}]
This result is a particular case of the following \relem{mat_vanish_fixsglin}, when we let $\sglin$ be the set of permutations under which the trajectories are invariant. In that case, the number of orbits of $\sglin$ is the number of different trajectories.
\end{proof}

\begin{lemma}[Matrices vanishing on $\fix(\sglin)$]
\label{lem:mat_vanish_fixsglin}
For any group of permutations $\sglin$, a basis of the space of matrices $Z$ vanishing on $\fix(\sglin)$ is given by the
$Z^{i, j, k}$'s introduced before \recor{synthesis}, where $i$ and $j$ are in the same orbit under $\sglin$.
Consequently, the dimension of this space is 
$$(N-1) \# \{\text{orbits under }\sglin\}.$$
\end{lemma}

\begin{proof}

Let $Z$ be such a matrix.
As $Z = \mdpl - \mdpl'$ with $M, M'$ diffusion matrices, the columns of $Z$ have vanishing sums. 
Moreover, $Z$ has to vanish on any vector fixed by $\sglin$. These vectors are precisely the $x \in \mathbb{R}^N$ such that $x_i = x_{\sigma(i)}$ for all $i\in\nods$, and for all $\sigma \in \sglin$. Thus, they are the $x$'s such that $E_{k, i} x = E_{k, \sigma(i)} x$ for every $i, k \in \nods$, and $\sigma \in \sglin$.
Therefore, $Z^{i, j, k} x = 0$ for all $x \in \fix(\sglin)$ when $i, j, k$ satisfy the assumptions of the lemma. The $Z^{i,j,k}$'s are clearly linearly independent. Let us show that they generate the space of all $Z$'s.

Let $Z$ vanish on $Fix(\sglin)$. We will make all of the coefficients of $Z$ vanish by substracting multiples of $Z^{i, j, k}$s, which will prove that Z is indeed a linear combination of the $Z^{i, j, k}$s. Let $\mathcal{O} = \{i_1 < ... < i_m\}$ be an orbit of $\{1, ..., n\}$ of cardinal $m$ under the action of $\sglin$. 

Let us remark that for $1 \leq k \leq N - 1$, $1 \leq l < m$, $Z^{i_l, i_{l+1}, k}$ satisfies the conditions of the lemma and has its $(k, i_l)$ coefficient equal to 1, its $(k, i_{l+1})$ coefficient equal to -1.

Therefore, $Z - Z_{k, i_1} Z^{i_1, i_2, k}$ has its $(k, i_1)$ coefficient equaling zero, and its $(k, i_2)$ coefficient equal to $Z_{k, i_2} + Z_{k, i_1}$, and aside from the last line (which we will ignore for the moment) these are the only coefficients changing.

Then, if $m \geq 3$, we can reiterate this by considering $Z - Z_{k, i_1} Z^{i_1, i_2, k} - (Z_{k, i_2} + Z_{i_1, k}) Z^{i_2, i_3, k}$ and the obtained matrix will have the $(k, i_2)$ coefficient vanishing and the $(k, i_3)$ coefficient changing to $Z_{k, i_3} + Z_{k, i_2} + Z_{k, i_1}$, and the $(k, i_1)$ coefficient is still 0. 

We iterate this method exactly $m - 1$ times to obtain $Z'$. By construction, $Z'$ has each of the $(k, i_l), 1 \leq l \leq m$ coefficients vanishing except maybe the $l = m$ one, equaling $Z_{k, i_1} + ... + Z_{k, i_m}$.

But then this one is also zero. Indeed, if $(e_j)_j$ is the canonical basis, as $\sum_{l = 1}^m e_{i_l} \in Fix(H)$, we have $Z' \sum_{l = 1}^m e_{i_l} = 0$, and by looking the $k$-th coefficient, we obtain $Z'_{k, i_m} = Z_{k, i_1} + ... + Z_{k, i_m} = 0$.

We can then iterate this construction on every line except the last (meaning for $1 \leq k \leq N - 1$) and every orbit to obtain $Z''$. By construction, every line of $Z''$ is zero, except maybe the last ($k = N$), but then as the columns of $Z''$ have a vanishing sum (as $Z''$ is a linear combination of $Z$ and the $Z^{i, j, k}$), $Z'' = 0$. Thus, $Z$ is in the space generated by the $Z^{i, j, k}$. 

\end{proof}

\end{document}